\documentclass{article}
\usepackage{arxiv}

\usepackage[utf8]{inputenc} % allow utf-8 input
\usepackage[T1]{fontenc}    % use 8-bit T1 fonts
\usepackage{hyperref}       % hyperlinks
\usepackage{url}            % simple URL typesetting
\usepackage{booktabs}       % professional-quality tables
\usepackage{amsfonts}       % blackboard math symbols
\usepackage{nicefrac}       % compact symbols for 1/2, etc.
\usepackage{lipsum}		% Can be removed after putting your text content
\usepackage{graphicx}
\usepackage[numbers]{natbib}
\usepackage{doi}
\usepackage{algorithm}
\usepackage{algpseudocode}
\usepackage{amsmath}
\usepackage{amsthm}
\usepackage{amssymb}
\usepackage{xcolor}
\usepackage{setspace}

\newtheorem{theorem}{Theorem}
\newtheorem{prop}{Proposition}

\newtheorem{conj}{Conjecture}
\newtheorem{definition}{Definition}
\newtheorem{defn}[definition]{Definition}

\newtheorem{claim}[prop]{Claim}

\newtheorem{corollary}[theorem]{Corollary}

\definecolor{mygreen}{RGB}{193,225,193}
\definecolor{myred}{RGB}{250, 160, 160}

\title{Rent Division with Picky Roommates}

\author{Yanqing Huang\\
	Department of Computer Science\\
        Harvard University \\
	Cambridge, MA 02138 \\
	\And
	Madeline Kitch \\
	Department of Applied Mathematics\\
        Harvard University \\
	Cambridge, MA 02138 \\
	\And
	Natalie Melas-Kyriazi \\
	Department of Computer Science\\
        Harvard University \\
	Cambridge, MA 02138
}

\renewcommand{\shorttitle}{}

\hypersetup{
pdftitle={Rent Division with Picky Roommates},
pdfsubject={q-bio.NC, q-bio.QM},
pdfauthor={Yanqing Huang, Madeline Kitch, Natalie Melas-Kyriazi},
pdfkeywords={First keyword, Second keyword, More},
}

\begin{document}
\maketitle

\begin{abstract}
How can one assign roommates and rooms when tenants have preferences for both where and with whom they live? In this setting, the usual notions of envy-freeness and maximizing social welfare may not hold; the roommate rent-division problem is assumed to be NP-hard, and even when welfare is maximized, an envy-free price vector may not exist. We first construct a novel greedy algorithm with bipartite matching before exploiting the connection between social welfare maximization and the maximum weighted independent set (MWIS) problem to construct a polynomial-time algorithm that gives a $\frac{3}{4}+\varepsilon$-approximation of maximum social welfare. Further, we present an integer program to find a room envy-free price vector that minimizes envy between any two tenants. We show empirically that a MWIS algorithm returns the optimal allocation in \textit{polynomial time} and conjecture that this problem, at the forefront of computer science research, may have an exact polynomial algorithm solution. This study not only advances the theoretical framework for roommate rent division but also offers practical algorithmic solutions that can be implemented in real-world applications.
\end{abstract}

% keywords can be removed

\section{Introduction}

Rent division is the study of how to assign rooms to tenants and fairly divide the rent among them. In the standard setting, the number of rooms is equal to the number of tenants. In order to ensure a division is ``fair," many rent division models satisfy a constraint called ``envy-freeness," which is the idea that every given participant prefers his or her assignment over that of everyone else \citep{velez_equitable_2018}. People have studied variants of this problem, including which notion of fairness to consider \citep{fairest}, what happens when budget constraints limit pricing \citep{procaccia_fair_2018}, complications from uncertainty in preferences \citep{peters2022robust}, and when people may be deciding among many apartments \citep{procaccia2024multiapartment}. 

In many practical cases, however, there may be more tenants than rooms, requiring some people to share with one another. We insist this is a consideration that requires new analysis. For instance, Anne might prefer to live in a room with Daniel than with Joe, and this preference impacts her utility, and consequently the total social welfare. It is not obvious that previous research can be applied to the rent division problem when it is extended to include roommates. 

While the classical setting relies on results relating to the allocation of goods among individuals \citep{alkan_fair_1991}, these tools no longer apply directly to the setting with roommates. One's consumption of a good (their place of living) is inherently tied to the other who consumes it. Moreover, the notion of envy-freeness is both ill-defined (how can one be envious of their roommate?), and, as we prove, a (roommate-excluded) envy-freeness solution does not always exist for maximum welfare assignments. Two papers on rent division address some components of this problem: ``Rental Harmony with Roommates" \citep{azrieli_rental_2014} and ``Assignment and Pricing in Roommate Market'' \citep{chan_assignment_2016}. The former introduces the rent division model with roommates and proves the existence of an envy-free solution when peoples' preferences do not depend on their roommate. The latter introduces the ``roommate market model," a variation on rental division where every room in a given house has two people, and people have separable utilities that depend on the price they pay, the room they receive, and roommate with whom they live.In this model, each tenant assigns a value to each room and to other tenants as roommates. Chan et al. show that maximizing the social welfare is NP-hard and thus propose a polynomial-time solution which achieves at least 2/3 of the maximum social welfare.

% Lastly, they apply the results of the ``assignment game" from a paper by Shapley and Shubik which proves one can find a pricing scheme that always achieves envy-freeness for a matching of individual agents to rooms \citep{citation-key}.

\begin{figure}[ht!]
    \centering
    \includegraphics[width=0.70\textwidth]{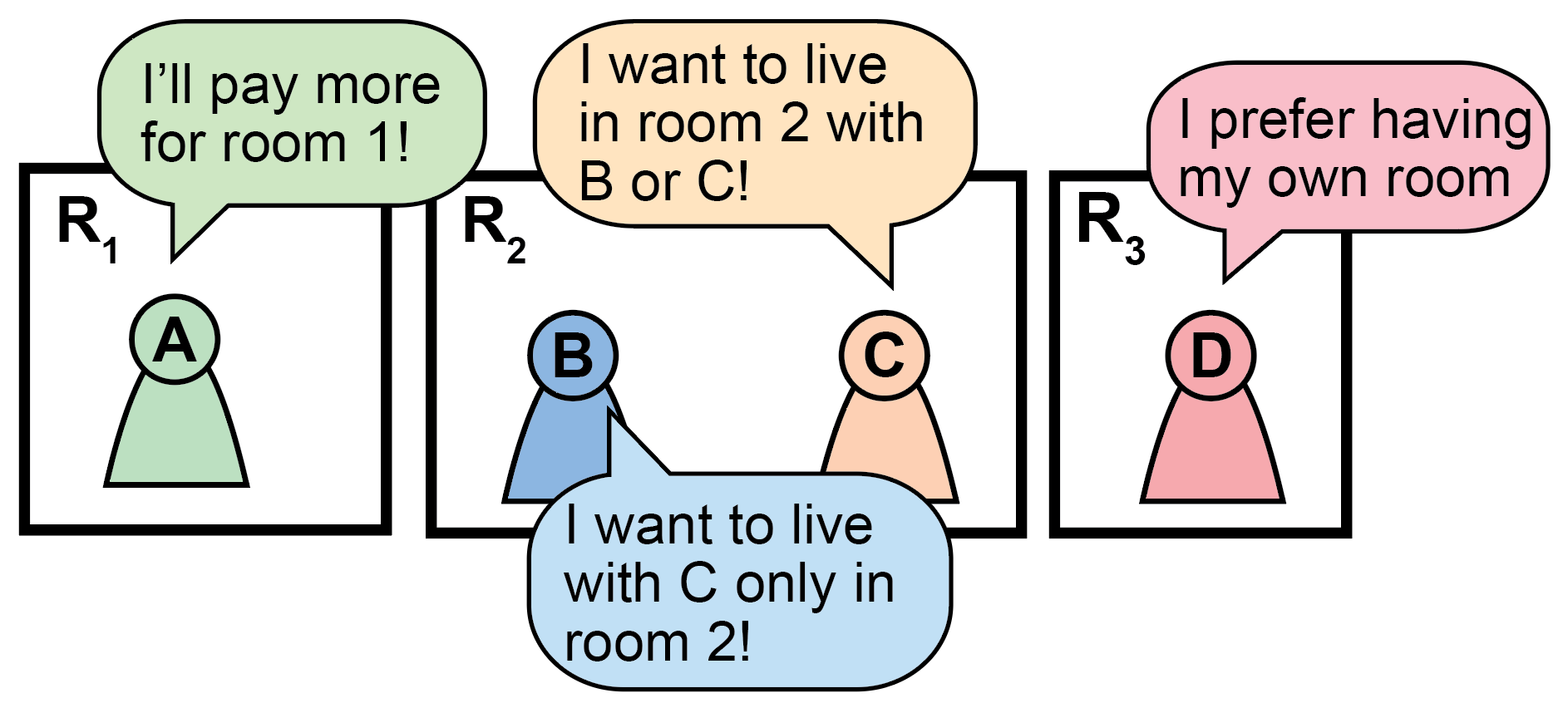}
    \caption{Illustrated Example of the Roommates Problem.}
\end{figure}

In this paper, we explore how to account for room and roommate preferences when room-sharing is only necessary for a subset of rooms in a house with non-separable utilities with regard to rooms and roommates. We present novel algorithmic approaches to (1) find the maximum social welfare assignment and (2) produce prices that minimizes tenant envy to ensure fairness.
%\footnote{For readers less familiar with this literature, we direct them to Lecture 11 ``Rent Division" from Harvard University's course titled Computer Science 238: Optimized Democracy (available online) for the relevant social choice framework and \citep{karthick_maximum_2017} for a relevant graph theoretic background.} 

\section{Model}
\subsection{Notation}

We approach this problem in three steps: first, we determine whether people will live alone, or if not, with whom they will share a room. We call these the rooming groups. The number of rooming groups is equal to the number of rooms. Second, we choose in what room the respective rooming groups will live. Lastly, we determine a set of room prices. Therefore, given an \textit{instance}, our goal is to output an \textit{assignment} consisting of rooming groups and their assigned rooms, and \textit{prices}.

\textit{Instance}. Let there be $m$ people and $n$ rooms with $m \in [n, 2n]$. People are indexed by $i \in [m]$ and rooms by $r \in [n]$. The \textit{instance} is defined by the tensor $V  \in \mathbb{R}_{+}^{m \times m \times n}$ and the total rent $R$. Component $v_{ijr}$ of $V$ is the valuation of person $i$ living with roommate $j$ in room $r$. Likewise, $v_{iir}$ is the utility for person $i$ living alone in room $r$. We assume utilities are quasi-linear so that a person $i$ assigned to live with roommate $j$ in room $r$ paying price $\tilde{p}_{r}$ has utility $u_i = v_{ijr}-\tilde{p}_{i}.$

\textit{Assignment}. In the assignment, we give each tenant a rooming group, and each rooming group a room. Rooming groups can be represented as a partition of the set of people, $[m]$, into $n$ disjoint subsets of size at most $2$. Formally, we are looking for a partition $M := \{S_g\}_{g \in [n]}$ such that for all $g$, $|S_g| \le 2$, and $\sqcup_{g=1}^{n}S_g = [m]$. Further, we define $\sigma: [n] \to [n]$ to be the permutation that maps rooming groups to rooms. 
% Since the number of rooming groups and rooms are the same, $\sigma \in \mathbb{S}_n$, the set of permutations from $[n]$ to itself. 
Thus, the assignment $A$ of people to roommates and rooms is given by the tuple $A = (M, \sigma)$. 
% Note that we wish to find the maximum social welfare assignment, which is the sum over every individual's respective utilities under the given assignment.

\textit{Prices}. Lastly, given $A$, we require a price vector $p \in \mathbb{R}^n$ where each component $p_{r}$ is the price of room $r$ such that the values sum to total rent $R$. Thus, any solution algorithm can be viewed as a map from the \textit{instance} $V$ to the \textit{solution} $(M,\sigma,p)$. Finally, we may divide rent among tenants, as represented by the vector $\tilde{p} \in \mathbb{R}^m$. 
\subsection{Key Definitions}
We seek a solution that results in the maximal happiness for the tenants, and guarantees some fairness or stability condition. This motivates the central definitions of our paper: the social welfare of an allocation and room envy-freeness. 

Define $V(S_g,r)$ to be the valuation of rooming group $g$ living in room $r$. This is given by
\[V(S_g,r) = \begin{cases}v_{ijr}+ v_{jir} & S_g=\{i,j\} \\ v_{iir} & S_g=\{i\}.\end{cases}\]

\begin{defn}
    We say that the social welfare of an assignment $A=(M,\sigma)$ for $M=\{S_g\}_{g\in [n]}$ is
    \[\text{SWF}(A) = \sum_{i \in [n]}V(S_g,\sigma(g))-R\]
\end{defn}
% Let $\mathcal{S}$ be the set of exact covers of $[m]$ containing either one or two elements (ie the set of feasible roommate pairs). 

We say that $A^* = (M^*, \sigma^*)$ is a maximum social welfare allocation if $A^* \in \mathrm{argmax}_{M, \sigma}\text{SWF}(A=(M,\sigma))$.

In order to assign rooming groups, we have to extend the definition of utility. For rooming group $S_g$ their utility is $U(S_g,r,p) =V(S_g,r)-p_r.$ The definition of room envy-freeness is that no rooming group would get higher total utility from swapping rooms with another rooming group and paying its respective price. This motivates our formal definition:

\begin{defn}
    We say that a \textit{solution} $(M,\sigma,p)$ for instance $V$ is \textit{room envy-free} (REF) if $\forall \ g,g' \in [n]$:
    \[U(S_g,\sigma(g),p_{\sigma(g)}) \ge U(S_g,\sigma(g'),p_{\sigma(g')}).\]
\end{defn}

% Additionally, we introduce a individual notion of fairness given by the worst-case proportional 
With regards to each tenant's envy, we define $\varepsilon$-EF as the following:
\begin{defn}\label{defn: e-EF}
     We say that a \textit{solution} $(M,\sigma, \tilde{p})$ for instance $V$ is $\varepsilon$-EF if $\forall i,j \in [n]$, $\{i, i'\} \in S_g$, $\{j, j'\} \in S_{g'}$:
    \[v_{ii'\sigma(g)} - \tilde{p}_i \ge \dfrac{1}{\varepsilon}(v_{ij'\sigma(g')} - \tilde{p}_j)\]
\end{defn}

In the case where $\varepsilon = 1$, tenant $i$ would rather live in their assigned room $\sigma(g)$ than switch with any other person $j$ and live with their roommate $j'$ while paying their price $\tilde{p}_j$. We call this person envy-freeness (PEF). We cannot provably guarantee PEF for welfare-maximizing allocations (Theorem \ref{thm:PEF}), and thus the relaxation to $\varepsilon > 1$ is needed for defining different degrees of envy in the rooming assignment. The higher $\varepsilon$ is, the higher the envy of the most envious tenant. Thus, lower values of $\varepsilon$ are generally associated with more ``fair'' solutions. 

\section{Assignment Algorithms}
The primary contribution of our theoretical work is to design algorithms that return allocations very close to or exactly equal to the welfare-maximizing allocations. 

As a brief aside, and to illustrate the problem's complexity, one might be interested to see if an exhaustive search is possible for small $(m,n)$. Why not just check all of the options? It turns out that brute forcing a maximum social welfare assignment is computationally intractable with complexity $O(2^nn!)$! Since the number of assignment options is the number of rooming groups times the number of possible matchings of groups to rooms, combinatorial formulas give us $\frac{\binom{m}{2n-m}\cdot  (2m-2n)!\cdot n!}{(m-n)!2^{m-n}}$ options. When $(m,n) = (4,3)$, there are $36$ assignment options. However, when $(m,n) = (12,9)$, there are more than $5$ billion options! Thus, even at relatively small scales an exhaustive search is impractical. 
\subsection{Greedy Algorithm}
As a first step, we propose and analyze a greedy algorithm to generate an assignment. The algorithm first generates all possible tuples $T = \{(i , j, r)\}$ where $i, j \in [m]$ and $r \in [n]$. Intuitively, this is all of the possible rooming group pairs in all possible rooms. Next, the algorithm calculates the utility of each tuple as the sum of the respective roommates' utilities for living with their roommate in room $r$. If this assignment is valid (if there are enough rooms left to accommodate all the unallocated people), the algorithm picks the tuple with the highest utility and adds it to the assignment $A$, and then removes every tuple which contains $i, j$, or $r$. If the assignment is not valid, it removes all singles from $T$ and repeats. This loop continues until all people and rooms are assigned. For an example, see section \ref{appendix: greedy ex}.
% See Appendix \ref{appendix: greedy ex} for a detailed example of the greedy algorithm.

\begin{algorithm}[ht!]
\caption{Greedy Algorithm}\label{alg:greedy}
\begin{algorithmic}[1]
\Require $m, n$
\State For every combination of roommates and rooms with $i \leq j$, create a set $T$ of tuples $(i, j, r)$ with valuation $v_{(i, j, r)} = v_{ijr} + v_{jir}$ if $i \neq j$ and $v_{ijr}$ if $i = j$.
\State Start with an empty assignment $A$.
\While{T $\neq \varnothing$}
\State Let the number of rooms yet to be allocated be $n'$ and number of people yet to be allocated be $m'$. 
\If {$m' = 2n'$}
    \State Remove all tuples in $T$ where $i = j$ (singles)
\ElsIf{$m' = n'$}
    \State Remove all tuples in $T$ where $i \neq j$ (doubles)
\EndIf
\State Pick the tuple $(i', j', r')$ maximizing $v_{(i', j', r')}.$
\If {$i' = j'$ and $m' - 1 > 2(n' - 1)$}
    \State Remove all tuples in $T$ where $i = j$ (singles)
    \State Skip to next iteration
\EndIf
\State Add tuple $(i', j', r')$ to assignment $A.$
\State Remove all tuples from $T$ that include $i$, $j$, or $r$, and increment $m'$, $n'.$
\EndWhile \\
\Return $A$
\end{algorithmic}
\end{algorithm}

\begin{claim}
This algorithm returns a valid assignment to the roommates problem.
\end{claim}
\begin{proof}
    We prove Claim 1 in a two step process: first, we show that every round of this algorithm returns a valid partial assignment, and then we show that this algorithm always terminates. Assume that we have a valid partial assignment $A_{\phi}$. For simplicity, we will use the $\{(i, j, r)\}$ notation for an assignment, but this can be easily translated to $A_\phi = (M_\phi, \sigma_\phi)$. Before a new tuple is added to the assignment, for all $(i, j, r) \in A_{\phi}$, $T$ does not contain any tuples with either $i$, $j$, or $r$, so the tuple $(i', j', r')$ added to the partial assignment from a round of the algorithm is pairwise disjoint from all the tuples in $A_\phi$, hence $A_{\phi'} = \{A_\phi, (i', j', r')\}$ is a valid partial assignment. Additionally, the algorithm terminates with $T$ being empty because before the last greedy round, the algorithm ensures that after $(i', j', r')$ is added to the allocation, $m' \leq 2 * n'$. In the last round, $n' = 1$ since each round removes a possible room, so $m' = 1$ or $m' = 2$, which means that we only have $1$ valid choice left in $T$, the picking of which terminates the algorithm. 
\end{proof}

\subsection{Greedy with Bipartite Matching}

Though the greedy algorithm returns a valid assignment to the roommates problem, it does not necessarily produce the assignment with the maximum social welfare with respect to that assignment. Thus, we introduce a post-processing step of bipartite matching to weakly increase the total social welfare produced from our greedy algorithm.

% Recall the example in Appendix \ref{appendix: greedy ex} but now consider when the utility for person 3 living alone in room 2 is 8. The greedy algorithm would still produce the assignment of $\{(3, 3, r_1), (1, 2, r_{2})\}$ with a total social welfare of 22, but the assignment $\{(3, 3, r_2), (1, 2, r_{1})\}$ would have produced a social welfare of 23. 

% For the sake of contradiction assume the greedy algorithm does not terminate. If it does not terminate that means that there must not be 

\begin{algorithm}
\caption{Bipartite Matching Algorithm}\label{alg:bma}
\begin{algorithmic}[1]
\Require $m, n, A.$
\State For every rooming group $S_g \in M$ and every room $r \in [n]$, create a complete bipartite graph between $\{S_g\}$ and $\{r\}$ with edge weights $V(S_g, r)$.
\State Find the maximum weight bipartite matching $\sigma^\star.$ \\
\Return $A^\star = (M, \sigma^\star).$
\end{algorithmic}
\end{algorithm}

This algorithm runs in polynomial time when implemented with the Hungarian algorithm \citep{hungarian}. Since $|M| = n$ and there are $n$ rooms, we know that every vertex will have degree $n$, thus the bipartite graph defined above is $n$-regular. It is well-known that using Hall’s Theorem, one can show that every regular bipartite graph has a perfect matching \citep{matchings}. The matching can be computed in time $O(n\log{n})$ \citep{goel2010perfect}. Note that the algorithm produces the maximum social welfare with respect to the valid matching $M$. As we show later, this implies that a REF price vector exists (\ref{thm:REF_matching}).

\subsection{Maximum Weight Independent Set (MWIS)}
The Maximum Weight Independent Set problem is defined on a graph $G = (X, E, w)$ and weight function $w: X \to \mathbb{R}_{+}$. The goal is to find the set $I \subseteq X$ such that $I$ has no vertices with edges between them and the total vertex weight $\sum_{x \in I}w(x)$ is maximized. We will show that there is a reduction of the MWIS problem to the max-welfare rental assignment problem. 

\begin{theorem}
For $m = 2n$, the maximum social welfare assignment is equivalent to a solution to the maximum weighted independent set problem on a graph $G$.
\end{theorem}
\begin{proof}
For the rent assignment setting let $X$ be the set of $(i, j, r)$ pairs for $i, j \in [m]$, $r \in [n]$, and $i < j$. To define the edges we assign the vertex $x \in X$ a set of three attributes $i$, $j$, $r$ and denote $|x_1 \cap x_2|$ as the number of attributes that $x_1$ and $x_{2}$ have in common. For example, $|(1, 2, r_0) \cap (2, 4, r_5)| = 1$ because they share the common attribute person $2$. Let $E$ be the set of \textit{edges}, where $e = (x_\ell, x_{\ell'}) \in E$ if $|x_\ell \cap x_{\ell'}| > 0$, meaning two nodes share at least a common roommate or room. Finally, we define the vertex weight function $w: X \to \mathbb{R}_{+}$ as the roommate valuation corresponding to $(i,j,r)$ so that $w(x) = v_{ijr} + v_{jir}$. 

Let $I$ be the maximum weight independent set of $G$. Note that $I$ has exactly $n$ nodes. Assume for the sake of contradiction that $I$ has $n + k$ nodes for some positive integer $k$. We see that independence would imply that there are $n+k$ distinct rooms, and we reveal a contradiction. Now assume $k$ is a negative integer, meaning $I$ has fewer than $n$ nodes. We know that there exists at least one free room and $2$ free tenants corresponding to a node $x \in X$ with positive weight by assumption, so again we reach a contradiction. As $I$ has size $n$ and all nodes are unconnected, the set must cover all $n$ rooms and all $2n = m$ potential roommates, representing a valid roommate assignment. 

As $I$ has the maximum vertex weight across all independent sets, we know that it maximizes $\sum_{x \in I}w(x)$ for those of size $n$. Note that this is exactly equal to the total tenant living valuations by construction. As rent stays the same, $I$ corresponds to the max welfare assignment. For an example, see section \ref{appendix: MWIS_alg_example}.
 
\end{proof}
As displayed in Figure \ref{fig2}, the central idea is that the nodes are all possible pairs of roommates and rooms, Furthermore, living arrangement valuations are represented by node weights and feasibility restrictions (someone can't live two places and one room can have one set of roommates) are represented by edges. Hence an independent (unconnected) set of nodes is a (partial) assignment. 

Note that Figure 2 displays the scenario in which $4$ people are assigned to $2$ rooms. Hence the MWIS has size $2$. We see that the combination $(1,2,r_2)$ and $(3,4,r_1)$ maximizes the node weight among non-adjacent edges, thus representing the welfare-maximizing assignment.
\begin{figure}[H]
    \centering
    \includegraphics[width=0.6\linewidth]{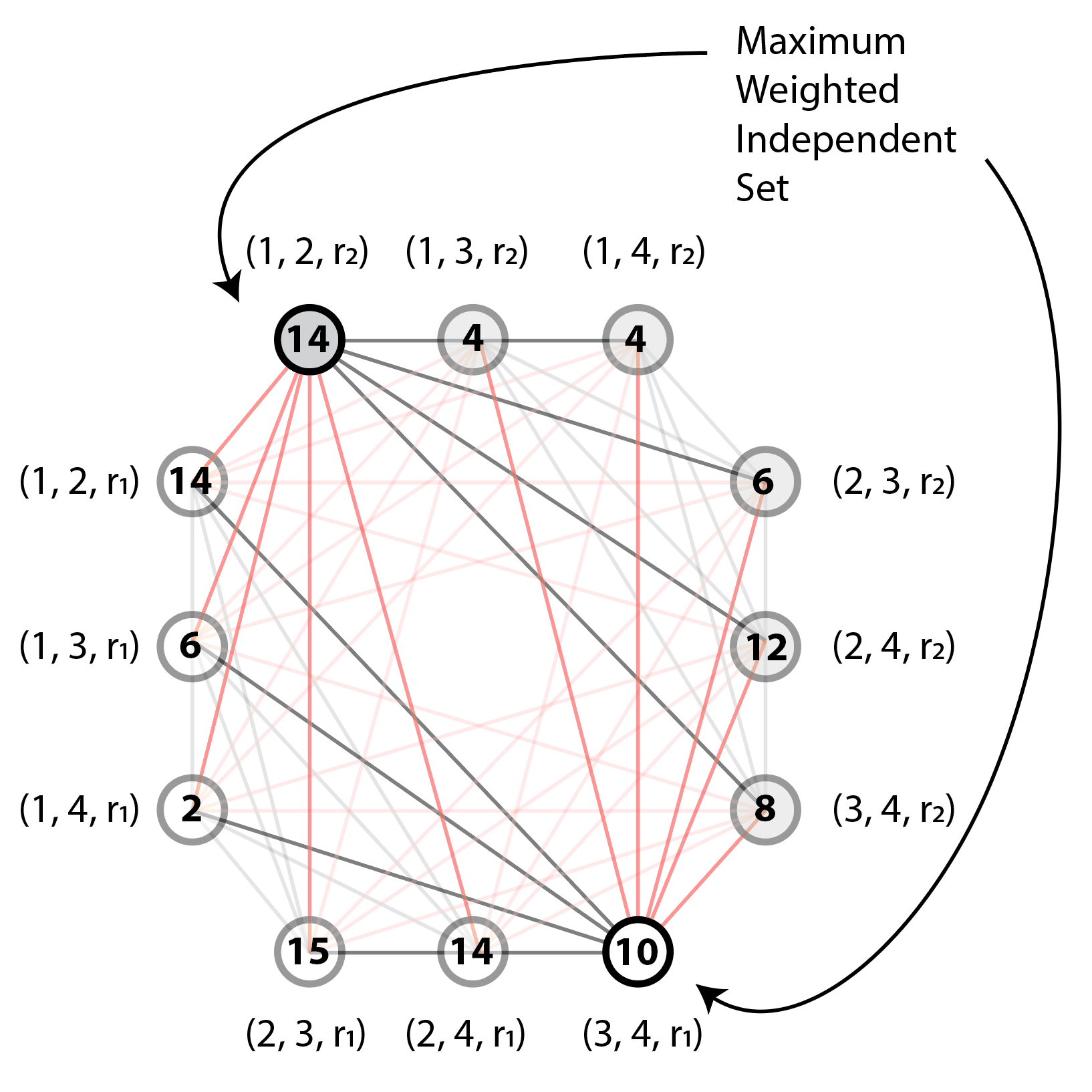}
    \caption{Maximum Weight Independent Set with Four Tenants and Two Rooms.}
    \label{fig2}
\end{figure}

\begin{theorem}\label{thm: MWIS}
For $m=2n$ and any $\varepsilon > 0$, there is a polynomial-time $(4/3 + \varepsilon)$-approximation algorithm for maximum social welfare assignment.
\end{theorem}

\begin{proof} Firstly, we will show that $G$ is 4-claw free. A graph $G$ is $d$-claw-free if a $d$-claw is not an induced subgraph of $G$. This means that there is no vertex that is connected to $d$ different vertices that are not connected to one another. 

% \begin{lemma}
%     $G$ is 4-claw free. 
% \end{lemma}
Next, assume for the sake of contradiction that a $4$-claw induced subgraph exists. Then $x$ is connected to $\{x_\ell\}_{\ell = 1}^{4}$ and all $x_\ell$ are pairwise disjoint. Since each $x_\ell$ is connected to $x$, $|x\cap x_\ell|\ge 1$. As $x_{\ell},x_{\ell'}$ are not connected, $|x_{\ell_1}\cap x_{\ell_2}|=0$. Therefore:
\[\left|\left(\cup_{\ell = 1}^{4}x \right)\cap x_\ell\right|=\sum_{\ell=1}^{4}|(x \cap x_\ell)| \ge 4> |x|.\]
This last step is a contradiction, since the cardinality of a set weakly decreases under intersections. As such, $G$ is 4-claw-free.  

% See Appendix \ref{appendix: proof lemma 9} for proof.

% \begin{proposition}
%     For perfect graphs, there exist an algorithm to find the MWIS solution in linear time.
% \end{proposition}

% \textcolor{blue}{Use proposition + lemma to show that from $G$, we can find a MWIS independent set in linear time; graph construction takes $O(n^3)$ time, so finding max welfare for $m = 2n$ is polynomial time.}

Cygan \citep{6686187} proposed a polynomial time approximation algorithm for MWIS in $d$-claw-free graphs with approximation ratio $d/3 + \varepsilon$. Note that with respect to the rooming groups $M = \{S_g\}$, we can always apply algorithm \ref{alg:bma} to generate a max social welfare assignment with respect to $\{S_g\}$ and guarantee room envy-freeness. Note that construction of the graph $G$ takes polynomial time. Therefore, using the maximum weighted independent set construction, we can find a $(4/3+\varepsilon)$-approximation algorithm for maximum social welfare assignment.

% \textbf{\textit{TODO: add introduction to algorithm}}
\end{proof}
\begin{algorithm}
\caption{MWIS Algorithm For Roommates Problem}\label{alg:mwis}
\begin{algorithmic}[1]
\Require $m = 2n$
\State For every combination of roommates and rooms with $i < j$, create a set of vertices $V$ consisting of tuples $(i, j, r)$ and vertex weight $V(S_g = \{i, j\}, r).$
\State Add edges between vertices $v_{\ell_i}$ and $v_{\ell_j}$ if $|v_{\ell_i} \cap v_{\ell_j}| > 0$.
\State Find the \textbf{maximum weighted independent set} on the graph $G = (X, E).$ \\
\Return Allocation $A$ corresponding to the assignments given by the MWIS.
\end{algorithmic}
\end{algorithm}

\begin{corollary}
    For $m < 2n$, there is a polynomial time $(4/3 + \varepsilon)$-approximation algorithm for maximum social welfare assignment.
\end{corollary}
\begin{proof}
    To reduce this to the setting of $m=2n$, we set the first $m$ individuals $1 \le i \le m$ to be real people, and the last $2n-m$ to be placeholders which do not correspond to any specific individual, called ``ghosts." Further, we redefine a new preference tensor $\tilde{V}$ as follows:
    \[\tilde{V}_{ijr} = \begin{cases} 
    V_{ijr} & i,j \le m \\
    \frac{V_{iir}}{2} & i \le m, j > m\\
    \frac{V_{jjr}}{2} & i > m, j \le m\\
    0 & \text{otherwise}. 
    \end{cases} \]
    There are three options for a given rooming group: it is made up of (1) only real people, (2) a mix of ghosts and people, and (3) only ghosts. In the first case $w(x) = V_{ijr} + V_{jir}$ as before. In the second case $w(x) = \frac{V_{iir}}{2} + \frac{V_{iir}}{2} = V_{iir}$ is the valuation of tenant $i$ living alone in room $r$. In the final case, $w(x) = 0$. Since ghosts living together have no impact on the weight of $I$, they will only do so if it is welfare-maximizing to leave one or more rooms unoccupied. Although it is possible in case (3) that $I$ won't be a rental assignment, we can simply fill the remaining rooms with the ghosts left over. As ghosts have no utility of living together or alone, the total weight is unchanged. Hence, replacing $\tilde{V}$ with $V$, we can directly apply Theorem \ref{thm: MWIS}, giving us the desired result.  
\end{proof}

% \textcolor{blue}{ENTER THE GHOSTS!!}
% The construction and proof of this corollary is in Appendix \ref{appendix: corollary}.

\section{Fairness Guarantees}

In general, we break down the assignment problem into two discrete components: (1) finding an assignment that maximizes social welfare, and (2) ensuring that the tenant prices make the assignment as fair as possible. Here we provide the guarantee of roommate envy-freeness and present an integer program to minimize envy among tenants. 
\vspace*{-0.02in}
\subsection{Room Envy-Freeness (REF)}
Recall that person-envy freeness (PEF) is guarantee that no tenant is envious of any other tenant in the apartment. We first note that PEF is too strong for a general fairness concept:
\begin{theorem}\label{thm:PEF}
    The existence of a person envy-free price vector for the max-welfare assignment is not guaranteed.
\end{theorem}
\begin{proof}
    Consider the following counterexample where room preferences don't matter, $m=2, n=4$, and rent is zero. Thus we are only looking for a envy-free roommate matching with prices. Let the preference tensor be defined as follows:
\[
V_{..r} = \begin{pmatrix}
0 & 12 & 2 & 8 \\
3 & 0 & 6 & 6 \\
2 & 6 & 0 & 11 \\
8 & 8 & 1 & 0 \\
\end{pmatrix} \quad \forall r
\]
where entry $(i,j)$ for $i \neq j$ denotes the utility of person $i$ living with person $j$. The max-welfare allocation is roommate pairs $(1,4)$ and $(2,3).$ The EF conditions for person $1$ envying person $3$ are $v_{14}-p_4 \ge v_{12}-p_2$ which implies $p_4 \le p_2-4$. Conversely, for person $3$ not envying person $1$ we have $v_{32}-p_2 \ge v_{34}-p_4$ which implies $p_2 + 5 \le p_4$. For both conditions to hold we would have $p_2+5 \le p_4 \le p_2 - 4$. Hence, an envy-free price vector for the welfare maximizing allocation does not always exist. 
\end{proof}
\vspace*{-0.1in}
Because of this, we relax our fairness guarantee to that of room envy-freeness (REF), and show that it is always attainable. 
\begin{theorem}\label{thm:REF_matching}
    For any matching $M=\{S_g\}_{g \in [n]}$, there exists price vector $p \in \mathbb{R}^n$ such that for the welfare-maximizing room assignments $\sigma^* \in \mathrm{argmax}_{\sigma \in \mathbb{S}_n}\sum_{r \in [n]}U(S_r;\sigma(r))$, $(\sigma^*, p)$ is room envy-free.  
\end{theorem}

\begin{proof}
\vspace*{-0.03in}
    Note that the utilities of each rooming group are quasi-linear and that REF is equivalent to envy-freeness for the rooming group. Therefore, with ``tenants'' $g \in [n]$ with utilities $U(S_g, \sigma(g),p_{\sigma(g)})=V(S_g,\sigma(g))-p_{\sigma(g)}$, it suffices to find a solution $(\sigma, p)$ to the individual rent-division problem. When (rooming group) utilities are quasi-linear, a (rooming group) EF solution always exists and corresponds to a welfare-maximizing allocation \citep{fairest}. 
\end{proof}

\begin{corollary}\label{corr:REF_max_SWF}
For any maximum social welfare assignment, $A^* = (M^*, \sigma^*)$, there always exists a price vector $p \in \mathbb{R}^n$ such that room envy-freeness is satisfied.
\end{corollary}
\begin{proof}
    Since $A^\star$ maximizes total utility, it also maximizes utility given a rooming assignment $M=\{S_g\}_{g \in [n]}$. 
    % Define rooming groups as agents, and their room (object) valuations as their total utility for rooming together in a the given room. By 
    A welfare-maximizing and REF solution $(\sigma,p)$ exists (Theorem \ref{thm:REF_matching}). By the Second Welfare Theorem, $(\sigma^*,p)$ is REF as well. 
\end{proof}

\begin{algorithm}[H]
\caption{$\varepsilon$-EF Algorithm}\label{alg:e-EF}
\begin{algorithmic}[1]
\Require $(M, \sigma)$ is a welfare-maximizing assignment
\State Construct new utilities matrix $V'\in \mathbb{R}^{m \times m}$ where $v'_{ij} = 0$ if $\{i, j\} \in M$ and $v'_{ij} = v_{ij'\sigma(g')}$ where $\{j, j'\} \in S_{g'}$. This means that player $i$'s valuation of player $j$ is the valuation of $i$ living with $j$'s roommate in $j$'s assigned room. Additionally, for notation, $v'_i$ represents the valuation of $i$ of its current assignment.
\State \textbf{Non-equal room-price sharing}: Compute a price vector $\tilde{p} \in \mathbb{R}^{m}$ by solving the linear program:
\begin{align*}
    \min\ &\varepsilon \\
    \text{s.t.: } &v'_{i} - \tilde{p}_i \geq \dfrac{1}{\varepsilon} (v'_{ij} - \tilde{p}_j) \quad\quad\quad \forall i, j \in [m] \\
    &\quad(\textbf{$\varepsilon$-EF constraint})\\
    &(v'_i + v'_{i'}) - (\tilde{p}_i + \tilde{p}_{i'}) \geq (v'_j + v'_{j'}) - (\tilde{p}_j + \tilde{p}_{j'}) \\
    &\quad \forall \{i, i'\} \neq \{j, j'\} \in M \quad\quad (\textbf{REF constraint})\\
    &\sum_{i = 1}^m \tilde{p}_i = 1 \quad\quad\quad \forall i \in [m]
\end{align*}
\State \textbf{Equal room-price sharing}: Let $r(i)$ be the room person $i$ is assigned to. Compute a price vector $p \in \mathbb{R}^{n}$ by solving the linear program:
\begin{align*}
    \min\ &\varepsilon \\
    \text{s.t.: } &v_{i}' - \frac{p_{r(i)}}{2} \geq \dfrac{1}{\varepsilon} \left(v'_{ij} - \frac{p_{r(j)}}{2}\right) \quad\quad\quad \forall i, j \in [m] \\
    &\quad(\textbf{$\varepsilon$-EF constraint})\\
    &(v'_i + v'_{i'}) - p_{r(i)} \geq (v'_j + v'_{j'}) - p_{r(j)} \\
    &\quad \forall \{i, i'\} \neq \{j, j'\} \in M \quad\quad (\textbf{REF constraint})\\
    &\sum_{r = 1}^n p_r = 1 \quad\quad\quad \forall r \in [n]
\end{align*}
\end{algorithmic}
\end{algorithm}

% Our proof does not rely on the size or rooming groups and is thus generalizable to settings with $3$ or more people living together.

\subsection{Fairness conditions given REF}
To decide among REF price vectors and split rent among roommates we rely on minimizing $\varepsilon$-EF (\ref{defn: e-EF}). Recall that this refers to ensuring that one's utility for rooming with someone else is at most $\varepsilon$ times what they currently receive given the max welfare assignment. 
% Thus decreasing $\varepsilon$ corresponds to decreasing potential envy. 

In practical settings, it is possible that roommates will be expected to share rent for a room equally. Even though envy of roommate is not well-defined, it may simply be considered unfair for two roommates to pay different prices. Therefore, in Algorithm 4, we present constructions with both equal and unequal splitting of rent among roommates, alongside the general option.

\section{Numerical Results}
To understand how the proposed algorithms and fairness guarantees work in practical situations, we simulated the roommate rent division problem in Python. We wish to quantitatively examine the performance of greedy and greedy with bipartite matching to MWIS, and whether the two envy-reducing price vector algorithms in \ref{alg:e-EF} produce similar fairness results.

\subsection{Greedy and Greedy Matching}
To simulate greedy and greedy matching algorithms, we first calculated the optimal solution (max social welfare) by using brute force. The runtime of brute force is exponential, making instances with large $n$ impractical to calculate (See \ref{appendix: brute}).

We first looked at how well the greedy algorithm performed against brute force for a range of $m$ and $n$ values below $n = 5$. In each simulation, we first generate the random utility matrices for each player so that $v_{ijr}$ are iid $\mathrm{Unif}(0,1)$. Next we calculated the social welfare of the assignment given by the greedy algorithm (see Section 4) as well as the greedy algorithm with bipartite matching (see Section 4.1). As shown in Figure 3, we see that the greedy algorithm produces an assignment that captures at least 90\% the maximum social welfare, while the greedy algorithm with bipartite matching improves the bound to 95\%.
\begin{figure}[H]\label{fig:Greedy_Greedy Match}
    \centering
    \includegraphics[width=0.60\textwidth]{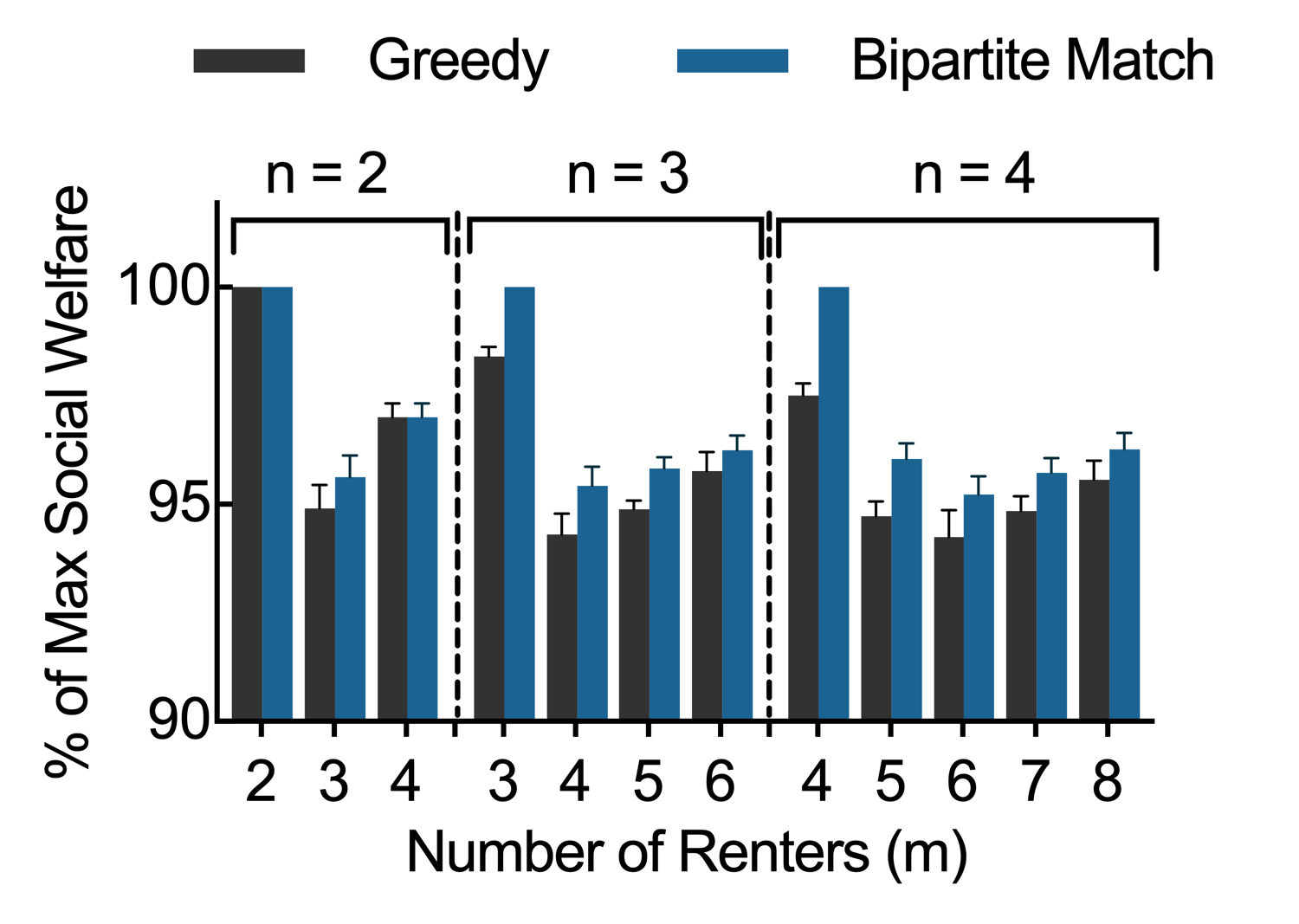}
    \caption{Mean percent of max social welfare generated by the greedy and greedy with bipartite matching algorithms for $10$ trials each of $200$ simulations.}
\end{figure}

\begin{figure}[H]\label{fig:Greedy_Greedy Match alpha}
    \centering
    \includegraphics[width=0.60\textwidth]{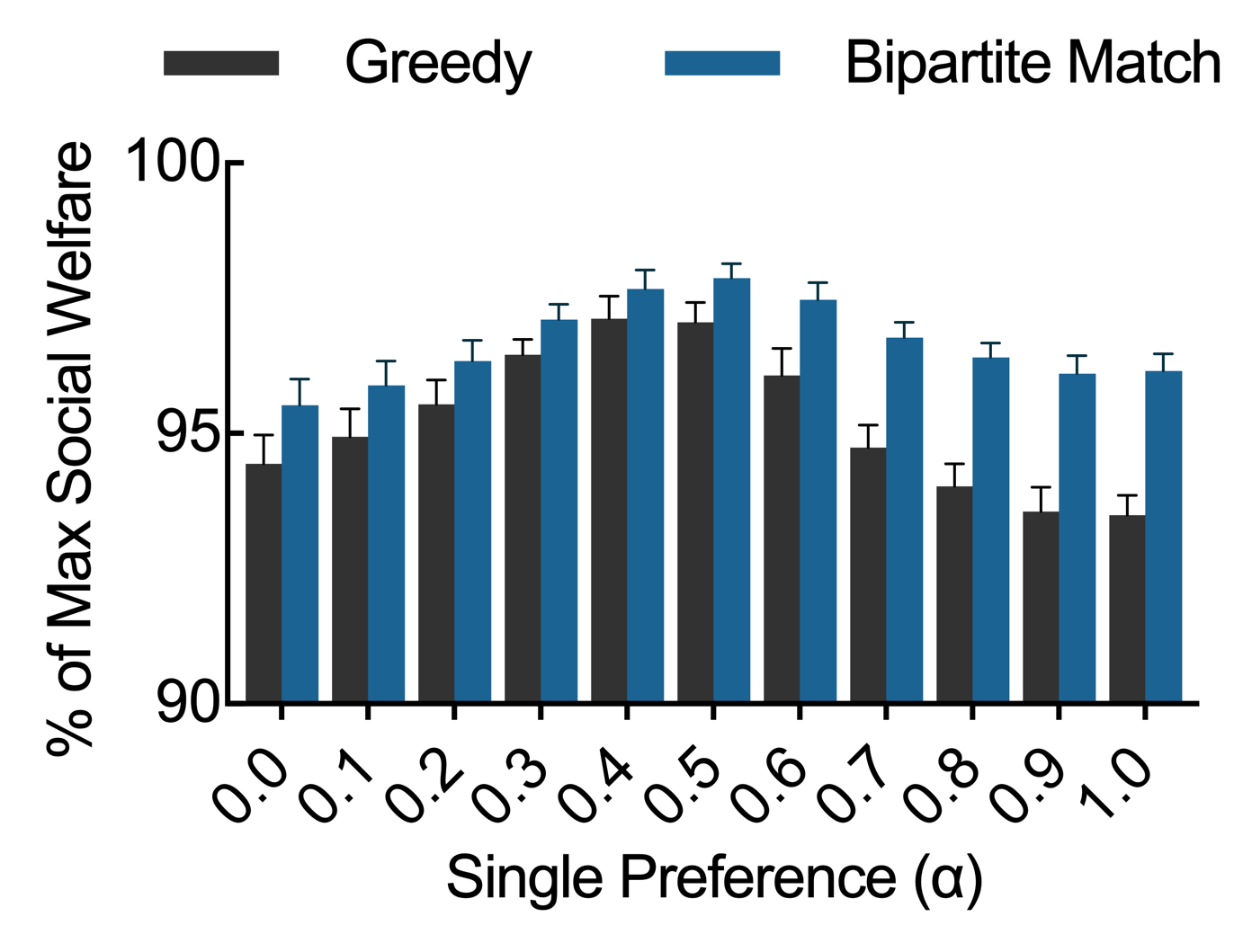}
    \caption{Mean percent of max social welfare generated for two greedy algorithms plotted against single preference for $50$ trials each of $200$ simulations.}
\end{figure}
However, in most real-life scenarios, valuation matrices are not exactly random. Often people have a preference for singles compared to doubles, and we wanted to see how our algorithm works when this is accounted for. We introduced a "single preference" parameter ($\alpha$) to the model, where a player would have $\alpha$ added to its valuations to all rooms where they live by themselves. Henceforth, we set $\alpha = 0.5$ so that $v_{iir}$ are iid $\mathrm{Unif}(0.5,1.5)$. Figure 4 shows the greedy and greedy with bipartite matching solutions with this parameter for $m = 4$ and $n = 3$. Note that bipartite matching consistently performs better than the greedy algorithm, but both algorithms' efficiency drops as $\alpha$ increases beyond $0.5$.

\subsection{MWIS}

Using an optimization package (Gurobi) in Python, we implemented the MWIS algorithm for the roommates problem. 
%Figure 4 shows the individual simulation results for a $m = 6$ and $n = 3$ instance. For visualization purposes, the simulations are ordered by the percent max social welfare of the greedy with bipartite matching algorithm. We see that MWIS always results in the optimal solution, while greedy + match gives the optimal solution about 45\% of the time, with minimum 70\% max social welfare.

% \begin{figure}[H]
%     \centering
%     \includegraphics[width=0.85\textwidth]{figures/MWIS Indiv.png}
%     \caption{Percent max social welfare for each simulation run of MWIS compared against greedy + bipartite matching for 1000 simulations. The data is sorted by the percent max social welfare for greedy match algorithm.}
% \end{figure}

\begin{figure}[H]
    \centering
    \includegraphics[width=0.7\textwidth]{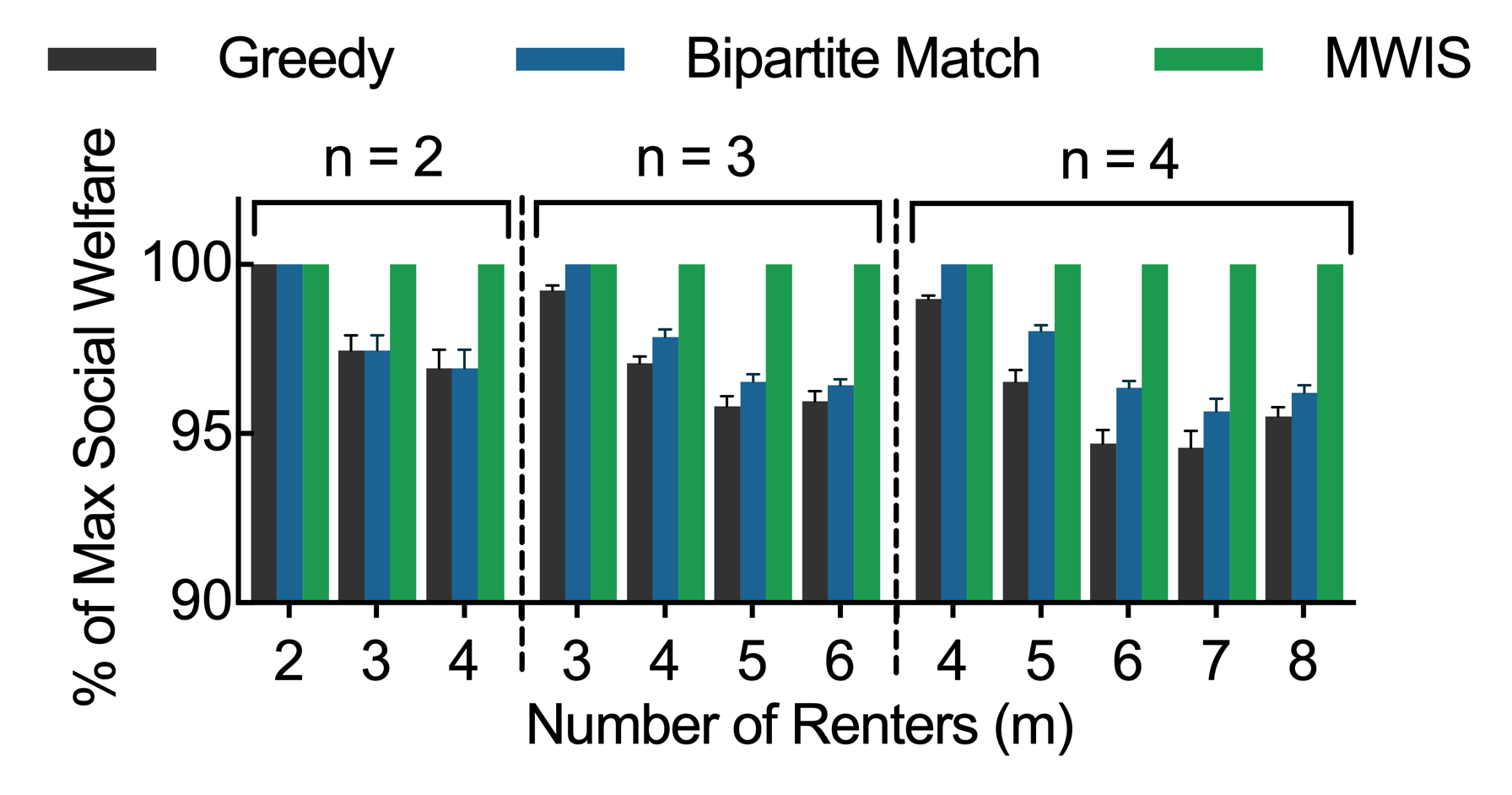}
    \caption{A comparison of greedy, greedy with bipartite matching, and MWIS algorithms. MWIS consistently outperforms both alternatives and returns a welfare-maximizing assignment when compared against brute force.}
\end{figure}
% Since we cannot enumerate all possible assignment combinations, we do not know if MWIS is optimal for $n > 5$, but regardless we observe that MWIS consistently outperforms greedy and greedy + bipartite matching.

As shown in Figure 5, we ran the algorithm on varied values of $n$. For $n \leq 5$, we were able to enumerate all the possible room combinations, plotting it against max social welfare, and see that MWIS always result in the optimal assignment whereas greedy and greedy with bipartite matching does not.

Furthermore, for small instances ($n \leq 12$), the algorithm runs in polynomial $O(n^3)$ time as shown in Figure 6 and runs relatively fast for real-world purposes. 
\begin{figure}[H]
    \centering
    \includegraphics[width=0.55\textwidth]{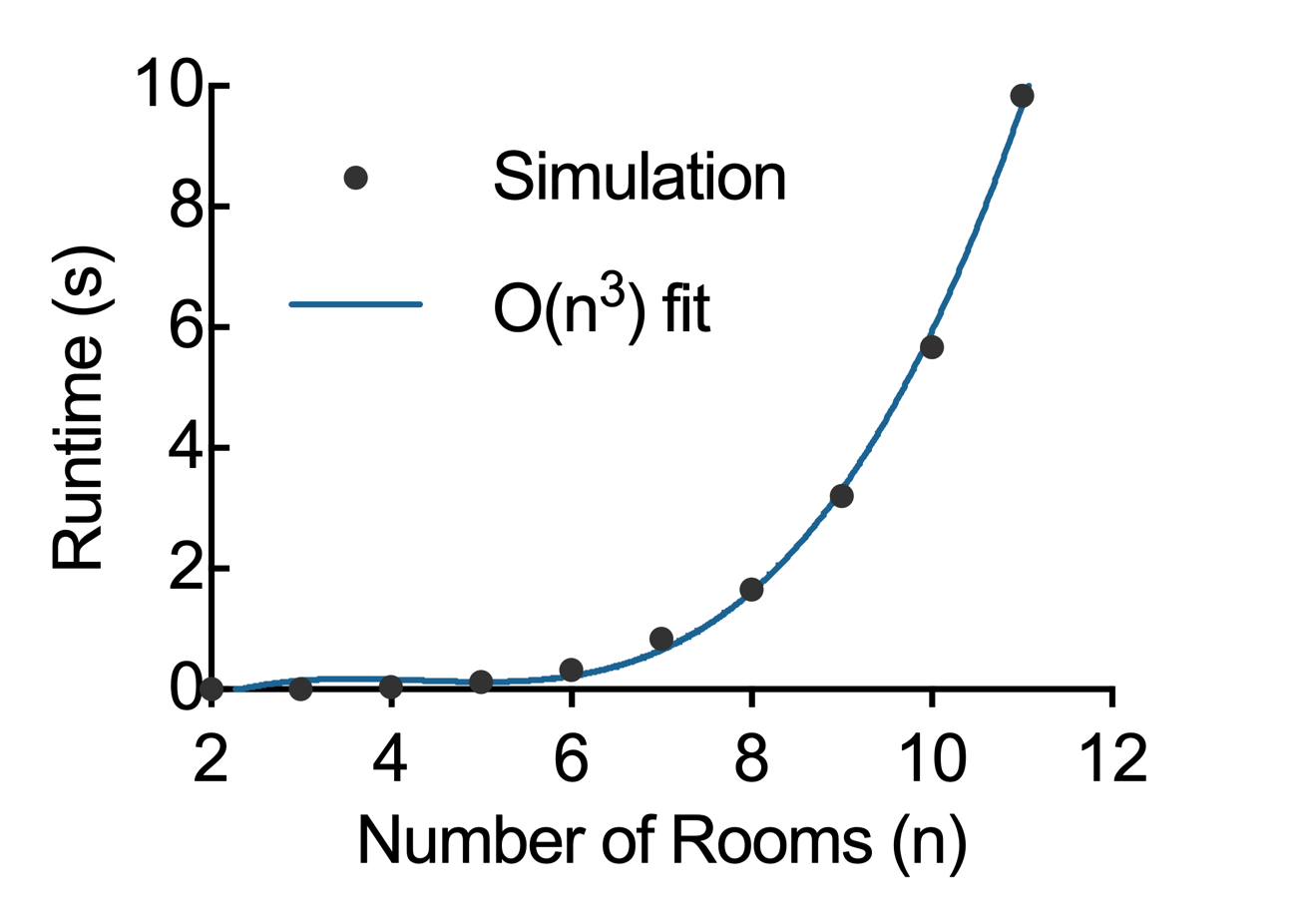}
    \caption{Runtime analysis of MWIS}
\end{figure}
However, the existence of an algorithm that finds the MWIS solution for graphs like $G$ in polynomial time is still an open question in graph theory. Related graphs such as claw-free graphs ($G$ is $4$-claw free) and perfect graphs have been shown to have polynomial time MWIS algorithms; this line of literature has overturned previous existing NP-hardness results. Therefore, we propose a conjecture that there exists such an algorithm for the set of graphs specified by roommate matching instances:

\begin{conj}\label{conj: polyMWIS}
There exists a polynomial time algorithm to find the maximum weighted independent set for the line graph of a 3-uniform hypergraph.
\end{conj}

% For a detailed discussion of this conjecture, see Appendix \ref{appendix: conjecture}

\subsection{Person Envy-Freeness}
% \textcolor{blue}{Note: make sure Defn of $\varepsilon$ aligns and also idk if we should use PEF or $\varepsilon$ PEF.}\\

Using greedy with bipartite matching, we analyzed whether person envy-free solutions exist. As we cannot provably guarantee PEF for all welfare-maximizing allocations (Theorem \ref{thm:PEF}), we plot the corresponding $\varepsilon$ value for the price assignment. Using simulation results with an integer program (Algorithm \ref{alg:e-EF}), we did not find any PEF solutions for both MWIS and greedy with bipartite matching. Over 50\% of all solutions were 4-EF for MWIS, and almost all solutions are 10-EF. We also observe that MWIS consistently outperforms greedy with bipartite matching in finding more fair solutions.

% Figure 7 shows the possibilities of having $\varepsilon$-EF solutions for random preferences in a $m = 6$, $n = 3$ scenario. We see that almost all instances are $2$-EF with the MWIS algorithm, while $65\%$ of all instances are PEF. 

%Additionally, the price vector is subject to room envy-freeness and with the maximin algorithm as detailed in \citep{fairest}.

\begin{figure}[ht!]
    \centering
    \includegraphics[width=0.50\textwidth]{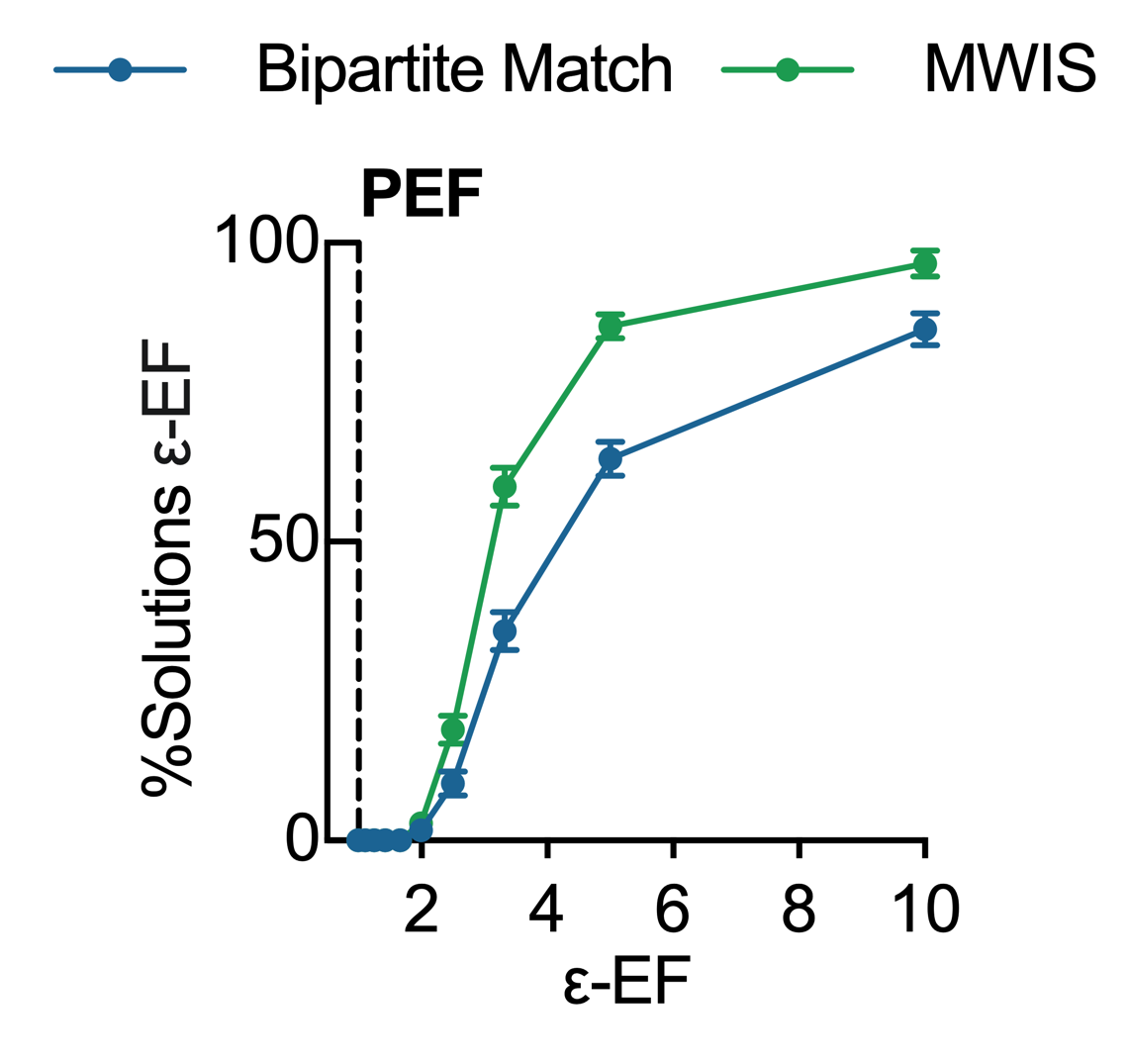}
    \caption{Numerical results of $\varepsilon$-EF for 10 trials of 200 simulations of $m = 6$, $n = 3$.}
\end{figure}

\begin{figure}[ht!]
    \centering
    \includegraphics[width=0.65\textwidth]{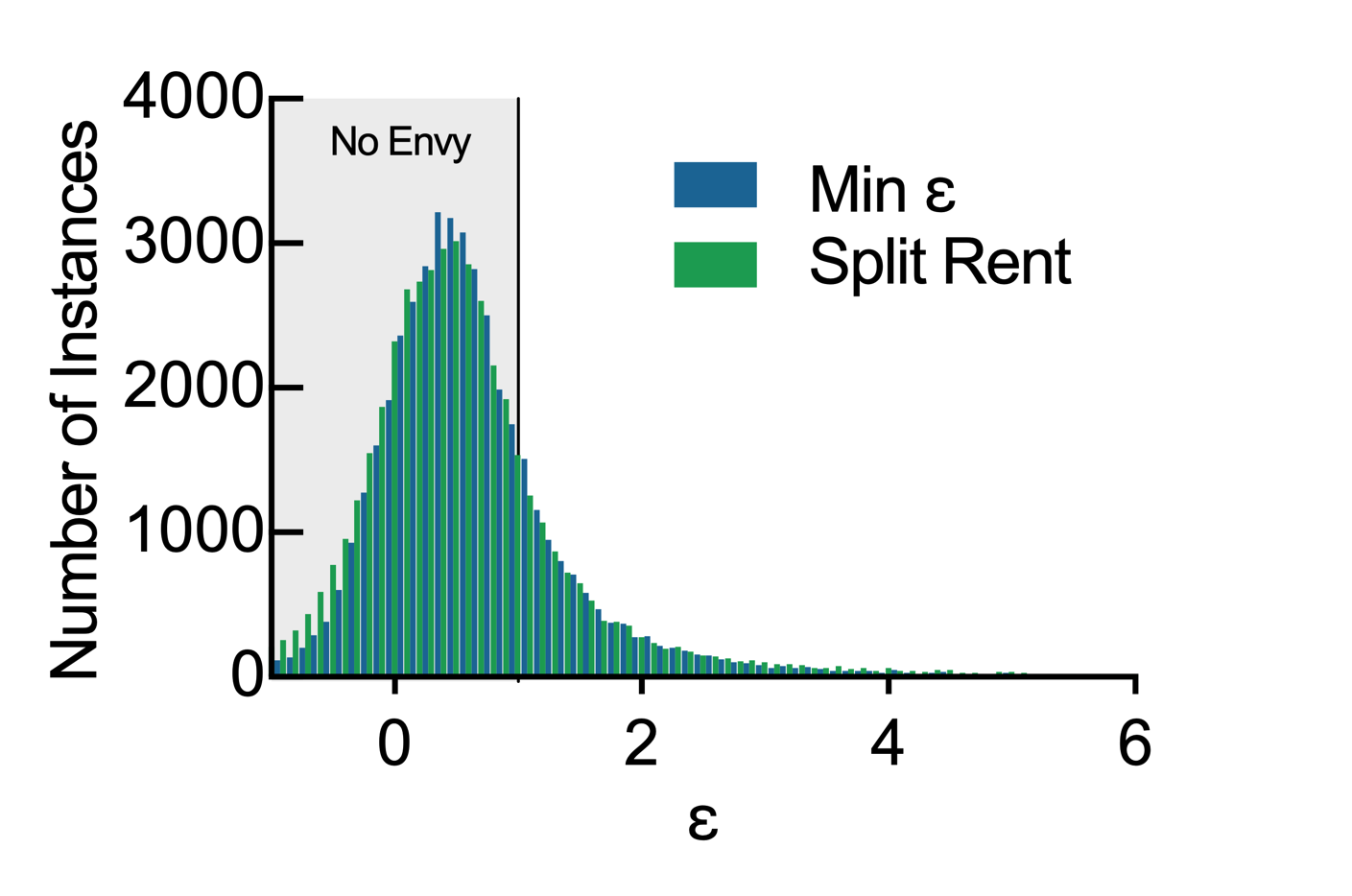}
    \caption{Numerical results of envy distributions for 2000 simulations of $m = 6$, $n = 3$.}
\end{figure}

Since PEF is unachievable globally, we looked at the distribution of individual envies with the two price vector algorithms given in \ref{alg:e-EF} (Figure 8). We note that tenants have no envy towards other tenants about 80\% of the time with the max social welfare assignment given by MWIS, and that splitting rent and min-$\varepsilon$ algorithms give approximately equally-fair solutions.

%We also looked at the minimum rent distributions of each simulation run's price vector subject to REF and PEF (Figure 8). We observe that the MWIS assignment algorithm also gives more ``fair" price vectors, as the distribution of the minimum rent that any player pays is higher than the price vector generated by greedy with bipartite matching.

\section{Conclusion}
In this paper, we have addressed the complex problem of fair rent division with tenants who have preferences over rooms and roommates. Our study extended traditional models of rent division to include scenarios where room-sharing is necessary for only a subset of rooms in the house. This is a common but under-explored situation in housing arrangements. For instance, our algorithms could be valuable for the selection of college roommates where living preferences are inferred from survey responses. 
% Existing literature has proposed variants of this problem, including models where tenants preferences between rooms and rommmates are disjoint \citep{chan_assignment_2016}, where tenant valuations are less truthful \citep{peters2022robust}, and where tenants have a choice between different potential apartments \citep{procaccia2024multiapartment}. 
For the joint-preferences scenario, we show that despite it appearing both computationally challenging and unlikely to provide some fairness guarantees, it is possible to provably find the maximum social welfare assignment for real-world cases, and room envy-freeness is always possible. Moreover, our algorithms could be applied to create assignments in other scenarios, such as the selection of projects (rooms) and partners (roommates) for a workplace assignment. More specifically, companies may want to maximize total productivity given than some individuals may work together better and have aptitudes for different projects.

There are many areas of future research. The most obvious one is working towards a definitive proof for Conjecture \ref{conj: polyMWIS}. Others include questions such as,
\begin{itemize}
    \item Is there a fairness notion stronger than REF that can be provably found? Potentially one could explore an $\varepsilon$ bound for $\varepsilon$-PEF under different preference distributions. 
    \item Under what settings do our results extend to more than $1$ roommate?
\end{itemize}

The strategies developed in this study offer promising solutions that balance fairness and efficiency. As housing dynamics evolve and the prevalence of shared living arrangements increases, the need for effective rent division algorithms will become ever more critical. Our work contributes foundational insights and tools that can adapt to these changes.

% \section{Acknowledgments}
% We'd like to thank Professor Ariel Procaccia and Lauren Conger for helpful advice and comments. 
\newpage
\bibliography{arxiv}

\newpage
\renewcommand{\thefigure}{A\arabic{figure}}
\setcounter{figure}{0}
\appendix
\section{Appendix}
\subsection{Discussion of Brute Force Algorithm}\label{appendix: brute}

We will see that the ability of an algorithm to do an exhaustive search is reasonable for small groups (a couple of friends living together), but scales very quickly. Note that $2n-m$ individuals will be alone, and $2m-2n$ individuals will have roommates. Therefore the number of options for rooming groups is the number of options for the subset assigned to singles, times the number of options of rooming groups among the remaining people:
\begin{align*}
    \text{\#\{rooming groups\} }=
    \frac{{m \choose 2n-m}\cdot (2m-2n)!}{(m-n)!2^{m-n}}
\end{align*}
Here we use that the number of ways to assign the remaining $2m-2n$ individuals to doubles is the number of way to assign $2m-2n$ to $m-n$ equal groups where order does not matter. Likewise, the number of ways to assign $n$ rooming groups to $n$ rooms is $n!$. Hence, an exhaustive search algorithm that looks over all possible rooming group and room pairs must evaluate
\[\text{\#\{rooming groups\} }\cdot \text{\#\{ways to assign groups to rooms\} } =\frac{{m \choose 2n-m}\cdot  (2m-2n)!\cdot n!}{(m-n)!2^{m-n}}\]
options. 
When $(m,n) = (4,3)$, there are $36$ assignment options. However, when $(m,n) = (12,9)$, there are more than $5$ billion options!

\subsection{Greedy Algorithm Example}
\label{appendix: greedy ex}

Let there be an apartment with three players and two rooms. Consider the following valuations for each player:

\begin{table}[ht!]
    \begin{tabular}[t]{|c| c | c|} 
     \multicolumn{3}{l}{Person 1 ($V_1$): } \\
     \hline
     Roommate & $r_1$ & $r_2$ \\ 
     \hline\hline
     1 & 10 & 5 \\ 
     \hline
     2 & 9 & 8 \\
     \hline
     3 & 2 & 2 \\
     \hline
    \end{tabular}
    \hfill
    \begin{tabular}[t]{|c| c | c|} 
     \multicolumn{3}{l}{Person 2 ($V_2$): } \\
     \hline
     Roommate & $r_1$ & $r_2$ \\ 
     \hline\hline
     1 & 6 & 8 \\ 
     \hline
     2 & 3 & 14 \\
     \hline
     3 & 2 & 5 \\
     \hline
    \end{tabular}
    \hfill
    \begin{tabular}[t]{|c| c | c|} 
     \multicolumn{3}{l}{Person 3 ($V_3$): } \\
     \hline
     Roommate & $r_1$ & $r_2$ \\ 
     \hline\hline
     1 & 4 & 2 \\ 
     \hline
     2 & 5 & 1 \\
     \hline
     3 & 6 & 6 \\
     \hline
    \end{tabular}
\end{table}

We can now generate the set of $T$ tuples $(i, j, r)$ and their respective utilities.

\begin{center}
\begin{tabular}{|c | c | c | c||} 
 \hline
 Roommate Pair & Room & Valuation \\ [0.5ex] 
 \hline\hline
 (1, 1) & $r_1$ &  10\\ 
 \hline
 (1, 1) & $r_2$ &  5\\ 
 \hline
 (2, 2) & $r_1$ &  3\\
 \hline
 (2, 2) & $r_2$ &  14\\
 \hline
 (3, 3) & $r_1$ &  6\\
 \hline
 (3, 3) & $r_2$ &  6\\
 \hline
 (1, 2) & $r_1$ &  9 + 6 = 15\\ 
 \hline
 \colorbox{mygreen}{$(1, 2)$} & \colorbox{mygreen}{$r_2$} &  \colorbox{mygreen}{$8 + 8 = 16$}\\
 \hline
 (1, 3) & $r_1$ &  2 + 4 = 6\\
 \hline
 (1, 3) & $r_2$ &  2 + 2 = 4\\
 \hline
 (2, 3) & $r_1$ &  2 + 5 = 7\\
 \hline
 (2, 3) & $r_2$ &  5 + 1 = 6\\
 \hline
\end{tabular} \\
\end{center}

Next the algorithm identifies all of the tuples that share a person or room with the chosen tuple. The former is in red and the latter is in green in the following table.

\begin{center}
\begin{tabular}{|c | c | c | c||} 
 \hline
 Roommate Pair & Room & Valuation \\ [0.5ex] 
 \hline\hline
 \colorbox{myred}{($1, 1$)} & $r_1$ & 10\\ 
 \hline
 \colorbox{myred}{($1, 1$)} & \colorbox{myred}{$r_2$} & 5\\ 
 \hline
 \colorbox{myred}{($2, 2$)} & $r_1$ & 3\\ 
 \hline
 \colorbox{myred}{($2, 2$)} & \colorbox{myred}{$r_2$} & 14\\ 
 \hline
 (3, 3) & $r_1$ & 6\\
 \hline
 (3, 3) & \colorbox{myred}{$r_2$} & 6\\
 \hline
 \colorbox{myred}{(1, 2)} & $r_1$ & $9 + 6 = 15$\\ 
 \hline
 \colorbox{mygreen}{(1, 2)} & \colorbox{mygreen}{$r_2$} & \colorbox{mygreen}{$8 + 8 = 16$}\\
 \hline
 \colorbox{myred}{(1}, 3) & $r_1$ &  2 + 4 = 6\\
 \hline
 \colorbox{myred}{(1}, 3) & \colorbox{myred}{$r_2$} &  2 + 2 = 4\\
 \hline
 \colorbox{myred}{(2}, 3) & $r_1$ &  2 + 5 = 7\\
 \hline
 \colorbox{myred}{(2}, 3) & \colorbox{myred}{$r_2$} &  5 + 1 = 6\\
 \hline
\end{tabular} \\
\end{center}

It then deletes all of the highlighted entries. This leaves only one remaining option:

\begin{center}
\begin{tabular}{|c | c | c | c||} 
 \hline
 Roommate Pair & Room & Valuation \\ [0.5ex] 
 \hline\hline
 \colorbox{mygreen}{(3, 3)} & \colorbox{mygreen}{$r_1$} & \colorbox{mygreen}{6}\\
 \hline
\end{tabular} \\
\end{center}

So, the algorithm add this option to the assignment list, removes this row from the table, and terminates. It outputs an assignment of players 1 and 2 living in room $a$ and player 3 living in room $b$ with a social welfare of 22. It is trivial to see that this is a valid assignment.

\subsection{MWIS Algorithm Example}\label{appendix: MWIS_alg_example}

We give an example for a graph with $m = 4$ and $n = 2$, as shown in Figure 7. The edges in black represent the edges between vertices with the same room ($r_1$ or $r_2$), while the edges in red represent the edges between vertices of different rooms. For clarity, the complement graph of $G$ showing which vertices are not connected is shown on the right of Figure 7.

\begin{figure}[ht!]
    \centering
    \includegraphics[width=1\linewidth]{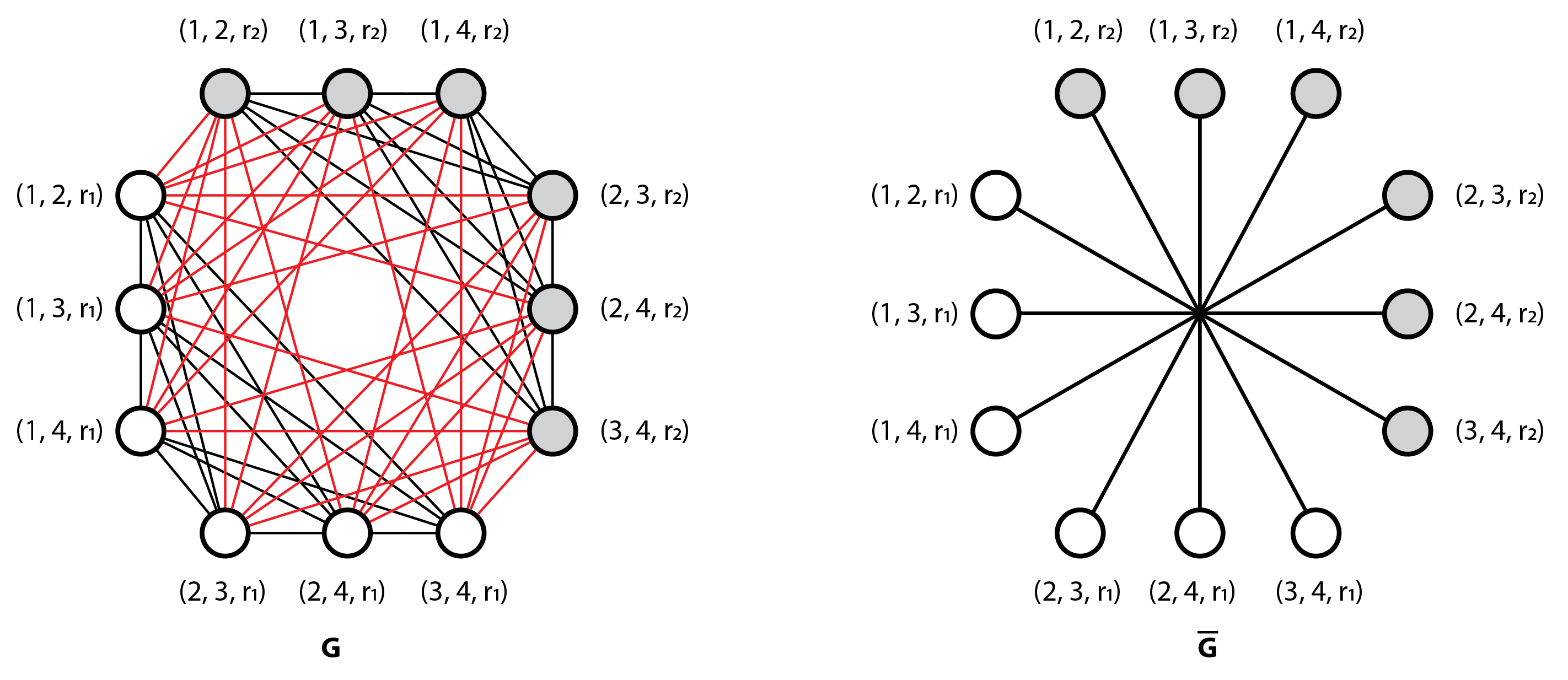}
    \caption{Illustration of the graph $G$ and its complement $\bar{G}$ for $m = 4$, $n = 2$.}
\end{figure}

Consider the following valuations for the players. 
\begin{table}[H]
    \footnotesize
    \begin{tabular}[t]{|c| c | c|} 
     \multicolumn{3}{l}{Person 1 ($V_1$): } \\
     \hline
     Roommate & $r_1$ & $r_2$ \\ 
     \hline\hline
     1 & 10 & 5 \\ 
     \hline
     2 & 8 & 8 \\
     \hline
     3 & 2 & 2 \\
     \hline
     4 & 1 & 1 \\
     \hline
    \end{tabular}
    \hfill
    \begin{tabular}[t]{|c| c | c|} 
     \multicolumn{3}{l}{Person 2 ($V_2$): } \\
     \hline
     Roommate & $r_1$ & $r_2$ \\ 
     \hline\hline
     1 & 6 & 6 \\ 
     \hline
     2 & 3 & 8 \\
     \hline
     3 & 7 & 5 \\
     \hline
     4 & 7 & 6 \\
     \hline
    \end{tabular}
    \hfill
    \begin{tabular}[t]{|c| c | c|} 
     \multicolumn{3}{l}{Person 3 ($V_3$): } \\
     \hline
     Roommate & $r_1$ & $r_2$ \\ 
     \hline\hline
     1 & 4 & 2 \\ 
     \hline
     2 & 8 & 1 \\
     \hline
     3 & 6 & 6 \\
     \hline
     4 & 6 & 5 \\
     \hline
    \end{tabular} 
    \hfill
    \begin{tabular}[t]{|c| c | c|} 
     \multicolumn{3}{l}{Person 4 ($V_4$): } \\
     \hline
     Roommate & $r_1$ & $r_2$ \\ 
     \hline\hline
     1 & 1 & 1 \\ 
     \hline
     2 & 7 & 6 \\
     \hline
     3 & 4 & 3 \\
     \hline
     4 & 8 & 9 \\
     \hline
    \end{tabular}
\end{table}       

We can construct the vertex weights like so, with the weight of each vertex equal to the valuation of the rooming group $\{i, j\}$. The maximum weighted independent set solution is highlighted, with a maximum social welfare of $14 + 10 = 24$.
\begin{figure}[H]
    \centering
    \includegraphics[width=0.55\linewidth]{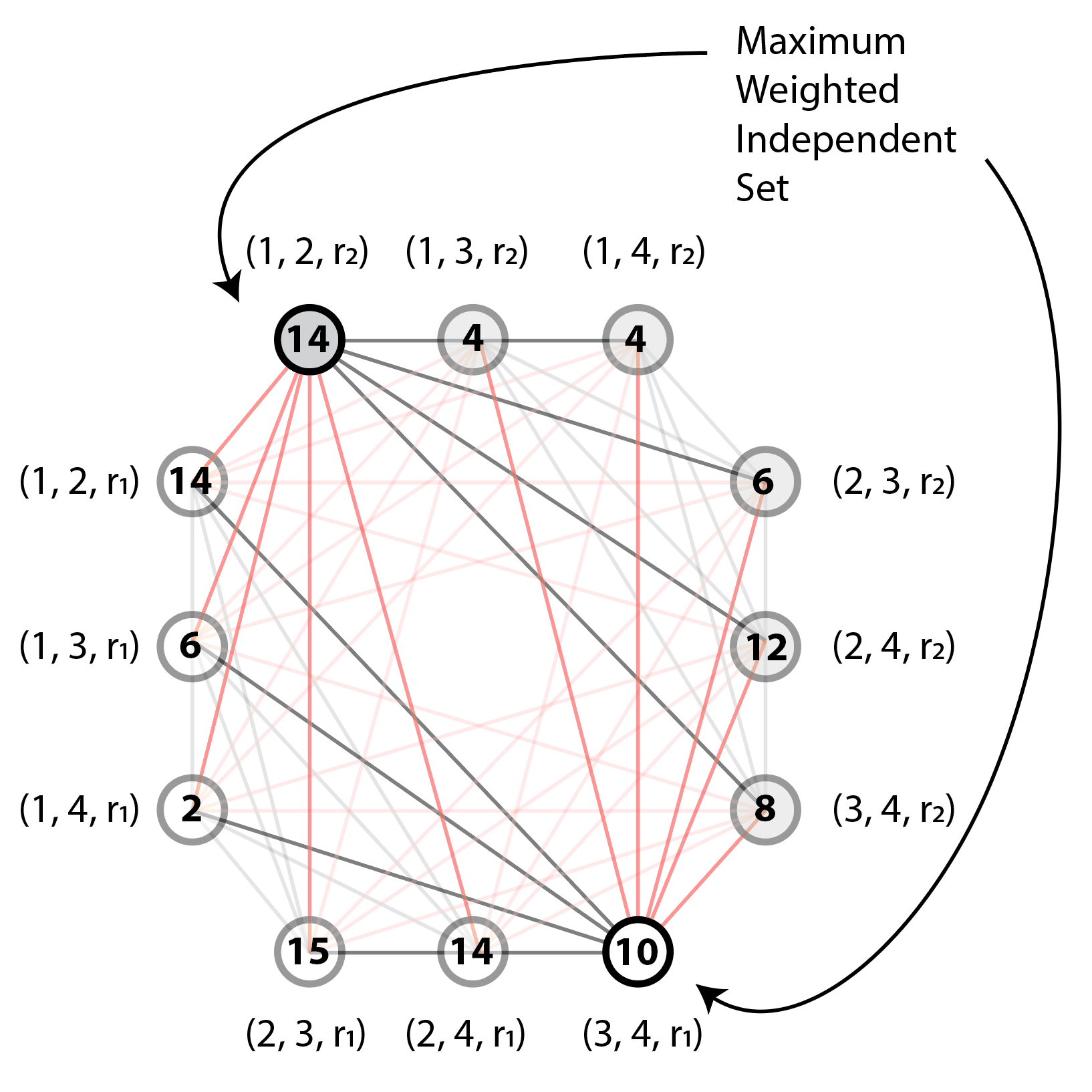}
    \caption{Maximum Weighted Independent Set with example valuation}
    \label{fig:enter-label}
\end{figure}

Compare with the greedy solution. The greedy algorithm would choose $(2, 3, r_1)$ on the first round since it has the highest valuation, and the only choice left in the second round is $(1, 4, r_2)$, which gives a social welfare of $15 + 4 = 19$, smaller than the maximum social welfare. Bipartite matching will also not be able to find the maximum welfare solution since the other choice $(2, 3, r_2)$ and $(1, 4, r_1)$ yields a social welfare of $6 + 2 = 8$.
                  
\subsection{Discussion of Conjecture 1} \label{appendix: conjecture}

In graph theoretic terms, the most accurate description of the MWIS graph $G$ is that it's a line graph of a 3-uniform hypergraph, as shown in Figure \ref{fig:hypergraph}. The hypergraph connects three sets of vertices, with each hypergraph edge intersecting 3 vertices: 1 in the set of roommate 1, 1 in the set of roommate 2, and 1 in the set of rooms.
\begin{figure}[H]
    \centering
    \includegraphics[width=0.95\linewidth]{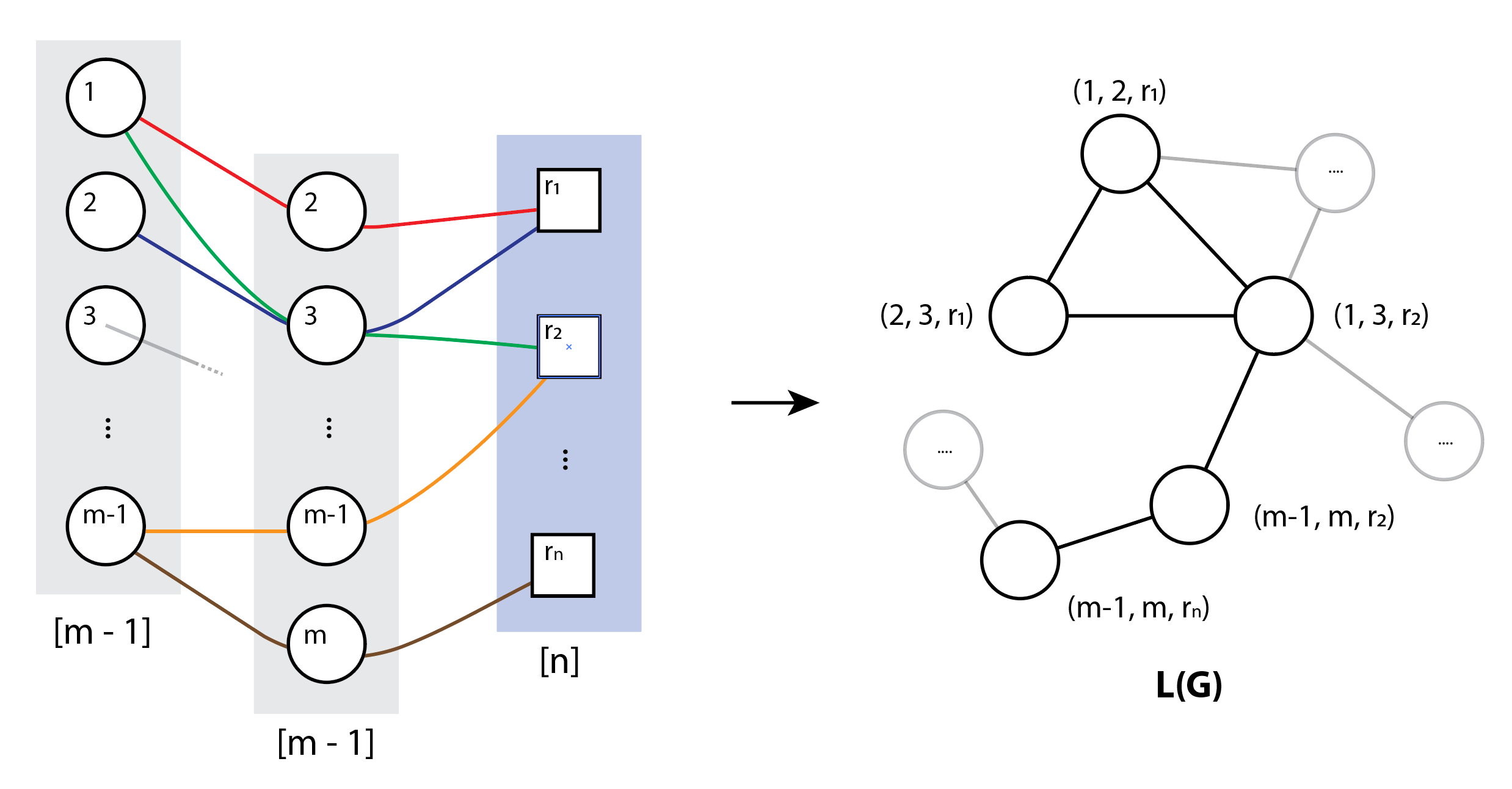}
    \caption{3-uniform hypergraph representation of the problem and its line graph}
    \label{fig:hypergraph}
\end{figure}

The line graph of a graph is the graph generated by turning every edge into a vertex and connecting two vertices if their edges in the original graph share a vertex. Therefore, we can think of each $(i, j, r)$ tuple as an edge in the original hypergraph, and two vertices in the line graph are connected if they have at least 1 attribute in common. 

Line graphs of regular graphs are perfect graphs, where MWIS can be solved in polynomial time \citep{combopt}. Additionally, new research in graph theory on the hardness of solving MWIS on graphs without $H$ as a subgraph have gotten closer to our graph in question. Dallard et al have shown that MWIS can be solved in polynomial time for $\{K_{1,d}, S, T\}$-free graphs, which are very close to our graph $G$. These novel results in graph theory and optimization gradually overturns previous NP-hardness results for special types of graphs, and thus it's very possible that MWIS can be solved for this graph in polynomial time, which would overturn the NP-hardness result in \citep{chan_assignment_2016}.





% \end{thebibliography}

\end{document}